\documentclass[letterpaper, 10 pt, conference]{ieeeconf}  

\IEEEoverridecommandlockouts                              
\usepackage[utf8]{inputenc}
\usepackage[T1]{fontenc}

\usepackage{cite}
\usepackage{amsfonts, amsmath, amssymb, bbm}
\usepackage{color}
\usepackage{comment}
\usepackage{import}
\usepackage{mathtools}
\usepackage[short]{optidef}
\usepackage{graphicx}
\usepackage{hyperref}
\usepackage{algorithm}
\usepackage{tikz}
\usepackage{booktabs}
\usepackage{graphicx}
\usepackage[table]{xcolor}

\definecolor{tablegray}{RGB}{244,246,248}
\definecolor{sectiongray}{RGB}{225,230,236}

\newtheorem{definition}{Definition}
\newtheorem{theorem}{Theorem}

\newtheorem{lemma}{Lemma}

\newtheorem{example}{Example}
\newtheorem{remark}{Remark}

\newcommand{\disturbance}{\mathbf{\widehat w}}
\newcommand{\operator}{\mathcal{M}}

\title{\LARGE \bf Context-Enriched Performance Boosting via Operator Decomposition}

\author{Leonardo Massai*, Sebastiano Messina**, Nicolas Kirsch*, and Giancarlo Ferrari-Trecate*%
\thanks{This research has been supported by the Swiss National Science Foundation under the NCCR Automation (grant agreement 51NF40\_180545) and the NECON project (grant number 200021-219431).}%
\thanks{* Authors with the Institute of Mechanical Engineering, Ecole Polytechnique Fédérale de Lausanne (EPFL), CH-1015 Lausanne, Switzerland (e-mail: l.massai@epfl.ch, nicolas.kirsch@epfl.ch, giancarlo.ferraritrecate@epfl.ch).}%
\thanks{** Author with the Department of Mathematical Sciences, Politecnico di Torino, IT-10129 Torino, Italy (e-mail: sebastiano.messina@polito.it).}%
}

\begin{document}

\maketitle
\thispagestyle{empty}
\pagestyle{empty}

%
\begin{abstract}
Performance Boosting (PB) is a control framework that, for a pre-stabilized
system subject to $\mathcal L_p$ process disturbances, parametrizes the
controllers that preserve closed-loop $\mathcal L_p$-stability through a
causal $\mathcal L_p$-stable operator mapping reconstructed disturbances to
corrective control actions. Although this permits optimization over expressive
stability-preserving controllers, learning a desired policy from disturbance
information alone can be difficult. We introduce a structured factorization
for context-enriched, multi-input PB operators. The proposed architecture
combines an $\mathcal L_p$-stable dynamical module that processes reconstructed
disturbances with a uniformly bounded matrix-valued mixer depending on
disturbances and contextual signals. Under the standard PB assumptions, this
factorization preserves closed-loop $\mathcal L_p$-stability by construction.
Moreover, on a weighted-envelope disturbance domain, we prove that the
factorization is necessary and sufficient for causal operators satisfying a
context-uniform envelope-preservation property. A numerical moving-gate
navigation experiment demonstrates the advantages of the proposed architecture
over context-agnostic PB, MAD, and reference-aware PB baselines.
\end{abstract}

\section{Introduction}

Modern control systems face increasing complexity in both their dynamics and objectives. In many modern applications, performance objectives depend explicitly on exogenous signals such as dynamic energy prices for battery management \cite{ali2025optimal} or real-time obstacle positions in mobile robotics \cite{10680398}. Traditional controllers, however, remain largely agnostic to these contextual inputs, limiting their effectiveness for sophisticated tasks requiring environmental awareness or anticipatory adaptation.

The expansion of sensor networks and data availability further emphasizes the need for control architectures that explicitly leverage exogenous information. Improving control performance through contextual awareness has therefore become an active research direction. By directly feeding exogenous signals to the controller, the system can better adapt to its operating environment and achieve more complex objectives. 

{Event-triggered control (ETC) provides a mechanism for reacting to online information \cite{tabuada2007event,heemels2012introduction}. The control action is updated when a condition on measured signals is satisfied, rather than periodically, with closed-loop stability certified via Lyapunov-based arguments. In most classical formulations, the triggering condition depends on endogenous signals such as the state or the deviation since the last sampling instant. The primary objective is to reduce communication and computation while preserving stability guarantees. Extensions to tracking settings \cite{tallapragada2013event} and networked control \cite{wang2011event} introduce exogenous signals, such as reference trajectories, but typically within the same stability-and-sampling framework, rather than as a general mechanism for performance optimization conditioned on rich contextual information. }

Several works have investigated this idea within Model Predictive Control (MPC). In \cite{10680398}, contextual information like the positions and velocities of obstacles is dynamically incorporated into an MPC framework for mobile robotics, improving efficiency and safety metrics. Other studies encode context implicitly through learned cost maps. For instance, \cite{goel2023semantically} learns dense navigation cost maps in which contextual relationships between objects and places are captured via spatially varying costs. In \cite{10920056}, nonlinear MPC is applied to car-following in autonomous driving, where the speed and position of the preceding vehicle are used as contextual inputs, enabling more precise tracking of the desired inter-vehicle distance. A popular framework in this direction is Perception-Aware MPC \cite{8593739}, where estimates of environmental states are incorporated into the MPC formulation, effectively parametrizing the optimization problem according to the current environment.

Although these approaches demonstrate strong empirical performance, providing closed-loop guarantees in the presence of time-varying contextual signals remains challenging. Naïvely incorporating such signals into costs and constraints may invalidate standard recursive feasibility and stability proofs in MPC. While certain guarantees can be established under specific assumptions in perception-aware MPC, deriving constructive procedures that ensure these properties remains difficult \cite{bonzanini2024perception}. The design of appropriate terminal costs and constraints becomes especially nontrivial beyond the linear setting.

Learning-based approaches have also emerged as powerful tools for context-aware control because of their inherent ability to process complex and large amounts of information. Such controllers are frequently designed to depend explicitly on exogenous signals such as environment measurements \cite{zimmerman2014neural} or fault indicators \cite{yen2000intelligent}. In \cite{huang2020learning} for example, a learning-based switching controller is proposed to ensure reliability of telecommunication services during faults.  Other works use neural networks to map contextual information to control parameters in parameter-varying systems. In \cite{kon2023directlearningparametervaryingfeedforward}, a time-varying scheduling signal is learned from input-output data, and a neural network is used to map this signal to controller gains. While often successful, the learning-based approaches presented here provide limited closed-loop guarantees.

The Performance Boosting (PB) framework introduced in \cite{10633771} \textcolor{black}{bridges the gap between neural network flexibility and closed-loop guarantees by optimizing over state-feedback policies that inherently preserve $\mathcal{L}_p$ stability}. In subsequent works \cite{11312988} and \cite{furieri2025mad}, \textcolor{black}{his framework was extended to reference tracking (rPB) and reinforcement learning (MAD), allowing the neural
network controller to incorporate specific exogenous inputs like reference signals or state trajectories, while retaining steady-state guarantees.} Promising results for fault compensation in drive systems were demonstrated in \cite{kirsch2025resilient}.
\textcolor{black}{While MAD and rPB successfully incorporated exogenous signals, their formulations were tailored to specific variables. They did not provide a generalized, multi-input, matrix-valued factorization for arbitrary contextual signals. Furthermore, while fields such as Adaptive Control and Event-Triggered Control focus heavily on adapting to changing parameters or environmental states, they generally lack the strict $\mathcal{L}_p$ stability guarantees for arbitrary non-$\mathcal{L}_p$ signals that the PB framework provides.}
This paper extends the PB framework to general contextual signals via a structured operator decomposition, yielding control policies that map process noise and contextual signals to corrective actions while preserving $\mathcal L_p$-stability. The proposed factorization isolates a stable disturbance-processing dynamical module from a uniformly bounded context-dependent mixing mechanism, providing an interpretable architecture that can simplify learning-based synthesis and improve closed-loop performance. Under mild conditions on the disturbances, we further establish that this factorization is necessary and sufficient. We also present a numerical experiment illustrating the resulting gains in a representative PB setting.

\section{Notation and Basic Definitions}
\label{sec:notation}

Vectors are denoted by lowercase letters, matrices by uppercase letters, and sequences by bold letters.
For a sequence $\mathbf{A}=(A_t)_{t\ge 0}$, we denote by $A_t^{i,j}$ the $(i,j)$-th entry of $A_t$. For every finite-dimensional space $\mathbb R^d$, we fix a vector norm
$\|\cdot\|$. For a matrix $A\in\mathbb R^{m\times s}$, the same symbol
denotes the corresponding induced norm,
$\|A\|
:=
\sup_{v\neq0}\frac{\|Av\|}{\|v\|}.$
Consequently,
$\|Av\|\leq\|A\|\,\|v\|,
\qquad
A\in\mathbb R^{m\times s},\quad v\in\mathbb R^s.$
\begin{definition}[$\ell^{n\times m}$ and $\ell_p^{n\times m}$ spaces]
Let $\ell^{n\times m}$ be the vector space of all matrix-valued sequences
$\mathbf{A}=(A_t)_{t\ge 0}$ with $A_t\in\mathbb R^{n\times m}$.
For $p\in[1,\infty)$, define \(
\|\mathbf{A}\|_{p} := \left(\sum_{t=0}^{\infty} \|A_t\|^{p}\right)^{\frac{1}{p}},
\) and let $\ell_p^{n\times m}:=\{\mathbf{A}\in \ell^{n\times m}:\|\mathbf{A}\|_p<\infty\}$.
Moreover, $\ell_\infty^{n\times m}:=\{\mathbf{A}\in \ell^{n\times m}:\|\mathbf{A}\|_\infty<\infty\}$ where
$\|\mathbf{A}\|_\infty := \sup_{t\ge 0}\|A_t\|$.
\end{definition}

\begin{definition}[Causal operator]
An operator $T:\ell^n\to \ell^m$ is causal if
$T(\mathbf{x}) = (T_0(x_0),\,T_1(x_{0:1}),\,\dots,\,T_t(x_{0:t}),\,\dots)$.
\end{definition}

\begin{definition}[$\mathcal L_p$-stable operator]
An operator $T:\ell^n\to \ell^m$ is said to be $\mathcal L_p$-stable if it is causal and
$T(\mathbf{x})\in \ell_p^m$ for all $\mathbf{x}\in \ell_p^n$. We write $T\in\mathcal L_p$.
\end{definition}
We now introduce the timewise matrix--vector product that will be used throughout the section to
define the proposed factorization.
\begin{definition}[Timewise matrix--vector product]\label{def:timewise-product}
Let $\mathbf{A}=(A_t)_{t\ge 0}\in \ell^{r\times s}$ and $\mathbf{v}=(v_t)_{t\ge 0}\in \ell^{s}$.
Define $\mathbf{A}\boxtimes \mathbf{v}\in \ell^{r}$ by
\(
(\mathbf{A}\boxtimes \mathbf{v})_t := A_t v_t,\quad t\ge 0.
\)
\end{definition}

\section{Performance Boosting framework}
\label{sec:pb}

We briefly recall the Performance Boosting (PB) framework and motivate the multi-input operator
studied in this paper. Consider a discrete-time system with dynamics\footnote{Any nonlinear system in the standard state-space form $x_{t+1}= \Tilde{f}_t(x_t,u_t)+w_t$ can be written as in \eqref{eq:system_state}. The interested reader is referred to \cite{10633771} for further details}: \begin{equation}\label{eq:system_state}
    x_t = f_t(x_{0:t-1},u_{0:t-1})+ w_t, ~~~t= 1,2,\ldots\,
\end{equation}
where $x_t \in \mathbb{R}^n$ is the state vector, $u_t \in \mathbb{R}^m$ is the control input and $w_t \in \mathbb{R}^n$ stands for unknown process noise. In operator form, system \eqref{eq:system_state} is: 
\begin{equation}
\label{eq:operator_form_state}
    \boldsymbol{x} = \mathbf{F}(\boldsymbol{x},\mathbf{u}) + \mathbf{w}\,,
\end{equation}
where $\mathbf F:\ell^n\times\ell^m\to\ell^n$ is the strictly causal
operator
\[
\mathbf F(\mathbf x,\mathbf u)
=
\bigl(
0,\,
f_1(x_0,u_0),\,
\ldots,\,
f_t(x_{0:t-1},u_{0:t-1}),\,
\ldots
\bigr).
\]
We adopt the convention
\[
w_0:=x_0,
\qquad
\mathbf w=(w_0,w_1,w_2,\ldots)
          =(x_0,w_1,w_2,\ldots),
\]
where $w_t$, for $t\geq1$, denotes the process disturbance. We assume
$\mathbf w\in\ell_p^n$. Over a finite horizon, disturbance realizations
$\mathbf w_{0:T}$ are sampled from a joint distribution
$\mathcal D_{0:T}$, whose time-$t$ marginal is supported on
$\mathcal W_t\subseteq\mathbb R^n$.

The main objective of PB is to design the control signal $\mathbf{u}$ such that it enhances the performance of the pre-stabilized system \eqref{eq:operator_form_state}, i.e. when $\mathbf{F}\in \mathcal{L}_p$, while preserving its steady-state guarantees\footnote{The steady state gurantees refer here to closed-loop stability, which means that the closed loop maps $\mathbf{w}\mapsto\mathbf{u}$ and $\mathbf{w}\mapsto\mathbf{x}$ are both in $\mathcal{L}_p$  }. To do so, PB relies on an internal-model-control (IMC) principle: a copy of the plant dynamics is used to reconstruct the disturbance signal
$\mathbf{\widehat w}=\mathbf{x}-\mathbf{F}(\mathbf{x},\mathbf{u})$
affecting the evolution of the plant. The boosting control action is then parametrized through an operator acting on this reconstructed signal:
$\mathbf{u}= \operator(\disturbance),$
where \(\operator\) is a causal \(\mathcal L_p\)-stable operator to be designed.

A central result of the PB framework \cite{10633771} is that, when the system model $\mathbf{F}$ is perfectly known, the search for performance-boosting, stability-preserving controllers can be
reduced without loss of generality to the search for suitable \(\mathcal L_p\)-stable operators \(\operator\). This is especially attractive from a learning viewpoint: one can optimize arbitrary performance costs over expressive classes of operators while retaining
closed-loop stability by construction. To keep the problem tractable and to solve it with unconstrained optimization, finite-dimensional and free parametrizations of $\mathcal{L}_p$ operators like Recurrent Equilibrium Networks (RENs)\cite{revay2023recurrent} or structured State Space Models (SSMs) \cite{11312999} are used. Given a control horizon $T$, the optimization problem solved by PB is: 
\begin{equation}\label{eq:opt_pb}
\begin{aligned}
\min_{\theta\in\mathbb R^{n_\theta}}
\quad&
\frac{1}{S}\sum_{s=1}^{S}
\mathcal J\bigl(x_{0:T}^s,u_{0:T-1}^s\bigr)
\\
\operatorname{s.t.}\quad
&x_0^s=w_0^s,
\\
&x_t^s
=
f_t(x_{0:t-1}^s,u_{0:t-1}^s)+w_t^s,
\quad t=1,\ldots,T,
\\
&\widehat w_0^s=x_0^s,
\\
&\widehat w_t^s
=
x_t^s-f_t(x_{0:t-1}^s,u_{0:t-1}^s),
\quad t=1,\ldots,T,
\\
&u_t^s
=
\mathcal M_t(\theta)(\widehat w_{0:t}^s),
\quad t=0,\ldots,T-1,
\end{aligned}
\end{equation}
where all constraints hold for $s=1,\ldots,S$ and
$\mathbf w_{0:T}^s\sim\mathcal D_{0:T}$. The parameter vector is
$\theta\in\mathbb R^{n_\theta}$, and $\mathcal J$ is any piecewise
differentiable loss defined over the realized state and
input trajectories. Problem~\eqref{eq:opt_pb} is typically solved using
backpropagation through time.

In the standard PB framework, the \(\mathcal{L}_p\)-stable operator $\mathcal{M}$ depends only on the reconstructed
disturbance \(\disturbance\). While in principle this allows recovering every stability-preserving controller,
producing a given stabilizing policy using only reconstructed disturbances can require complex operators \cite{furieri2025mad}. {A clear example is the imitation of a simple state-feedback law. Since the standard PB operator only observes the reconstructed disturbances, it would implicitly need to reconstruct the state by re-simulating the plant dynamics before being able to respond to it, a highly complex internal computation. If the state is instead supplied as an additional input, the operator only needs to learn the feedback response itself, a much simpler task.} Therefore, in many applications it is greatly beneficial to enrich the controller with
contextual information. \textcolor{black}{However, because contextual signals (such as state-dependent features or task parameters) are generally not $\ell_p$ sequences, directly feeding them into an operator bounded in $\mathcal{L}_p$ -spaces would fundamentally invalidate the theoretical closed-loop stability guarantees}. This motivates the study of an extended operator of the form
\begin{equation}\label{eq:extended_operator}
    \operator:\ell_p^n\times \ell^q \to \ell_p^m,
\qquad
(\disturbance,\mathbf z)\mapsto \mathbf{u},
\end{equation}
where 
\(\mathbf z\in\ell^q\) is a contextual sequence, \(\disturbance\in\ell_p^n\) the disturbance reconstruction, and such that \(\operator(\disturbance,\mathbf z)\in\ell_p^m\).

\section{Main Results}
\label{sec:main}
The main contribution of this paper is to show that extended operators like the one in \eqref{eq:extended_operator} can be realized by considering a structured factorization that separates
(i) a bounded (``\(\ell_\infty\)'') matrix-valued component depending on \((\disturbance,\mathbf z)\),
and
(ii) a disturbance-driven \(\ell_p\) component depending on \(\disturbance\) only.
The idea is to keep disturbance processing in a causal
\(\mathcal L_p\)-stable block \(\operator_p:\ell_p^n\to\ell_p^s\), inject context through a causal bounded mixer \(
\operator_\infty:\ell_p^n\times \ell^q\to \ell_\infty^{m\times s}.
\)
and compose the maps via the timewise product:
\begin{equation}\label{eq:op-decomp}
\operator(\disturbance,\mathbf z)
:= \operator_\infty(\disturbance,\mathbf z)\boxtimes \operator_p(\disturbance).
\end{equation}


Notice that the factorization \eqref{eq:op-decomp} generalizes the MAD policies \cite{furieri2025mad} as well as the rPB policies \cite{11312988}. Indeed, MAD policies
decompose the control input into an \(\mathcal L_p\)-stable magnitude term and a bounded direction term
depending on the state trajectory. In our framework, this is recovered as the special case \(s=1\),
where \(\operator_p(\disturbance)\) generates the scalar magnitude and \(\operator_\infty(\disturbance,\mathbf z)\) acts as
a bounded \(m\times 1\) mixer, with \(\mathbf z\) chosen as the state sequence (or a bounded function
thereof). The rPB framework corresponds to the special case where the mixer \(\operator_\infty(\disturbance,\mathbf z)\) is diagonal, and the context $\mathbf z$ is the reference to track. Our formulation is more general because it allows arbitrary contextual signals,
multi-dimensional disturbance features, and matrix-valued bounded mixing across multiple control
channels. Finally, the factorization does not prescribe a universal choice of the latent
dimension $s$. The abstract representation result holds for every $s\geq 1$;
thus, $s$ should be interpreted as a practical architectural hyperparameter
controlling the dimension of the disturbance features supplied to the mixer,
rather than as a requirement for representability. In the numerical experiment
of Section~\ref{sec:exp-gate}, performance varied only marginally over a broad
range of values of $s$, although the most appropriate choice may generally
depend on the task and on the selected finite-dimensional parametrizations.

{From a practical standpoint, both $\operator_p$ and $\operator_\infty$ admit simple finite-dimensional parametrizations suitable for unconstrained optimization. The module $\operator_p$ belongs to the same class of causal $\mathcal L_p$-stable operators as $\operator$ in the traditional PB setting, and can therefore be parametrized using RENs \cite{revay2023recurrent} or stable SSMs \cite{11312999}. As for $\operator_\infty$, any neural network with a bounded activation function (e.g., sigmoid or tanh) on its output layer produces a uniformly bounded output, regardless of its internal architecture, and thus satisfies the required constraint by construction.}

\subsection{Operator decomposition and \texorpdfstring{$\mathcal L_p$}{Lp}-stability}
We first show that the proposed factorization is always sufficient to preserve
\(\mathcal L_p\)-stability: the timewise product of a bounded mixer and an \(\ell_p\) disturbance
feature sequence remains in \(\ell_p\).
We use throughout the vector and induced matrix norms introduced in
Section~\ref{sec:notation}.

\begin{theorem}[Sufficient condition]
\label{thm:suff}
Let
\[
\operator_p:\ell_p^n\to\ell_p^s,
\qquad
\operator_\infty:
\ell_p^n\times\ell^q\to\ell_\infty^{m\times s}
\]
be causal operators. Define
\[
\operator(\disturbance,\mathbf z)
:=
\operator_\infty(\disturbance,\mathbf z)
\boxtimes
\operator_p(\disturbance).
\]
Then $\operator$ is causal and
\[
\operator(\disturbance,\mathbf z)\in\ell_p^m,
\qquad
\forall\disturbance\in\ell_p^n,\quad
\forall\mathbf z\in\ell^q.
\]
Moreover,
\begin{equation}
\label{eq:suff-pointwise}
\|\operator(\disturbance,\mathbf z)\|_p
\leq
\|\operator_\infty(\disturbance,\mathbf z)\|_\infty
\|\operator_p(\disturbance)\|_p.
\end{equation}
\end{theorem}

\begin{proof}
Causality follows from the causality of $\operator_p$ and
$\operator_\infty$, together with the timewise definition of
$\boxtimes$. Moreover, for every $t\geq0$,
\begin{align*}
\|\operator(\disturbance,\mathbf z)_t\|
&=
\left\|
\operator_\infty(\disturbance,\mathbf z)_t
\operator_p(\disturbance)_t
\right\| \\
&\leq
\|\operator_\infty(\disturbance,\mathbf z)_t\|
\|\operator_p(\disturbance)_t\|.
\end{align*}
Therefore,
\[
\sum_{t=0}^{\infty}
\|\operator(\disturbance,\mathbf z)_t\|^p
\leq
\left(
\sup_{t\geq0}
\|\operator_\infty(\disturbance,\mathbf z)_t\|
\right)^p
\sum_{t=0}^{\infty}
\|\operator_p(\disturbance)_t\|^p.
\]
Since
$\operator_\infty(\disturbance,\mathbf z)\in
\ell_\infty^{m\times s}$
and
$\operator_p(\disturbance)\in\ell_p^s$,
the right-hand side is finite. Hence
$\operator(\disturbance,\mathbf z)\in\ell_p^m$, and taking the
$p$-th root gives~\eqref{eq:suff-pointwise}.
\end{proof}

\subsection{Non-conservative factorization in specific noise regime}

Theorem~\ref{thm:suff} shows that the proposed factorization is always sufficient to preserve
\(\mathcal L_p\)-stability. We now study when the converse is true, namely when an extended operator
\(\operator\) can be represented in the form \eqref{eq:op-decomp}. This is not automatic on the full
domain \(\ell_p^n\times \ell^q\). We therefore restrict attention to a subdomain, in which disturbances are transient in a uniform weighted sense. On this
assumption, we will show that the proposed factorization is non-conservative. This regime is
characterized by two ingredients: (i) an admissible class of disturbance signals \(\mathbf w\),
and (ii) a corresponding regularity property of the operator \(\operator\). We start by introducing
the disturbance class.

\begin{definition}[Weighted \texorpdfstring{$\ell_\infty$}{linfty} envelope space]
\label{def:ellinfty-kappa}
Fix \(p\in[1,\infty)\) and let \(\kappa=(\kappa_t)_{t\ge 0}\) be a positive sequence such that \(\sum_{t=0}^{\infty}\kappa_t^{-p}<\infty.\) For \(d\in\mathbb N\), define
{\small
\begin{equation*}
    \ell_{\infty,\kappa}^d
:=
\left\{
\mathbf v\in \ell^d:\ \|\mathbf v\|_{\infty,\kappa}<\infty
\right\}, \quad  \|\mathbf v\|_{\infty,\kappa}
:=
\sup_{t\ge 0}\kappa_t\|v_t\|.
\end{equation*}
}
\end{definition}

The space \(\ell_{\infty,\kappa}^d\) collects sequences whose magnitude is uniformly dominated by the
time-varying template \(\kappa_t^{-1}\). Indeed, it is easy to show that \(\|\mathbf v\|_{\infty,\kappa}<\infty\) is equivalent to
the existence of a constant \(C<\infty\) such that
$\|v_t\|\le C\,\kappa_t^{-1} \: \forall t\ge 0,$
so \(\kappa_t^{-1}\) plays the role of a prescribed decay (or envelope) profile.
The summability condition on $\kappa$ ensures that this template is \(\ell_p\)-summable,
and thus any \(\mathbf v\in\ell_{\infty,\kappa}^d\) automatically belongs to \(\ell_p^d\), with the
explicit bound
\begin{equation}\label{eq:kappa-embed-inline}
\|\mathbf v\|_p
\le
\left(\sum_{t=0}^\infty \kappa_t^{-p}\right)^{\!\frac1p}
\|\mathbf v\|_{\infty,\kappa}.
\end{equation}
If two envelopes \( \kappa^{(1)} \) and \( \kappa^{(2)} \) satisfy
\( \kappa_t^{(1)} \le C\,\kappa_t^{(2)} \) for all \( t\ge 0 \) and some \( C>0 \), then
\( \ell_{\infty,\kappa^{(2)}}^n \subseteq \ell_{\infty,\kappa^{(1)}}^n \). Hence, "slower" envelopes
define larger classes, since they also contain all sequences with faster decay. At the same time,
for any fixed \( \kappa \), the class \( \ell_{\infty,\kappa}^n \) remains broad: because
\( \kappa_t>0 \) for all \( t \), it contains all finitely supported sequences and is therefore dense
in \( \ell_p^n \) for every \( p\in[1,\infty) \). Nevertheless, \( \ell_{\infty,\kappa}^n \) is in
general a strict subset of \( \ell_p^n \). For instance, if \( \kappa_t=(t+1)^\alpha \) with
\( \alpha>1/p \), the sequence \( w_t=2^{-j/p} \) when \( t=2^j \) for some \( j\in\mathbb N \), and
\( w_t=0 \) otherwise, belongs to \( \ell_p \) but not to \( \ell_{\infty,\kappa} \).

\smallskip

We now introduce the second ingredient of the specific regime studied, namely, a mild regularity condition on
the extended operator \(\operator\) over the admissible disturbance class.

\begin{definition}[Envelope-preserving operator]\label{ass:env}
Let \(\operator\) be causal on \(\ell_{\infty,\kappa}^n\times \ell^q\) with values in \(\ell_p^m\).
We say that $\operator$ is envelope-preserving if, for every \(\disturbance\in\ell_{\infty,\kappa}^n\),
\begin{equation}\label{eq:env}
\sup_{\mathbf z\in\ell^q}\,\sup_{t\ge 0}\kappa_t\,\|\operator(\disturbance,\mathbf z)_t\| < \infty.
\end{equation}
\end{definition}

Notice that \eqref{eq:env} is equivalent to requiring that, for every
\(\disturbance\in\ell_{\infty,\kappa}^n\), there exists a finite constant
\(C_{\operator}(\disturbance)>0\), \emph{independent of the context}, such that
\(
\|\operator(\disturbance,\mathbf z)_t\|
\le
C_{\operator}(\disturbance)\,\kappa_t^{-1},
\: \forall t\ge 0,\ \forall \mathbf z\in\ell^q.
\) The uniformity of the bound over \(\mathbf z\) is essential: it rules out
operators whose response can be delayed arbitrarily by the context
(cf.\ Example~\ref{ex:no-envelope} below), and it is exactly the
additional regularity that makes a bounded-mixer factorization possible on the regime studied, as
shown by the next result. Moreover, this definition is not merely abstract: a broad
family of dynamical operators commonly used in modern control applications is envelope-preserving, as we shall see later on.

Using Definition~\ref{ass:env}, we can establish a converse to
Theorem~\ref{thm:suff}: within the considered regime, the bounded-mixer factorization is not only
sufficient for \(\mathcal L_p\)-stability, but it is also \emph{necessary} for operators satisfying the
same uniform envelope regularity. 

\begin{theorem}[Necessary and sufficient factorization]
\label{thm:equiv}
Fix \(p\in[1,\infty)\), let \(\kappa=(\kappa_t)_{t\ge 0}\) be a positive sequence such that \(\sum_{t=0}^{\infty}\kappa_t^{-p}<\infty\)
and let \(s\ge 1\).
Let \(\operator\) be causal on \(\ell_{\infty,\kappa}^n\times \ell^q\). Then the following are equivalent:
\begin{enumerate}
\item[\textup{(i)}]
\(\operator\) is envelope-preserving, i.e., for every \(\disturbance\in\ell_{\infty,\kappa}^n\),
\begin{equation}\label{eq:env-iff}
\sup_{\mathbf z\in\ell^q}\|\operator(\disturbance,\mathbf z)\|_{\infty,\kappa} < \infty.
\end{equation}

\item[\textup{(ii)}]
There exist causal operators
\[
\operator_p:\ell_{\infty,\kappa}^n\to \ell_{\infty,\kappa}^s,
\qquad
\operator_\infty:\ell_{\infty,\kappa}^n\times \ell^q
\to \ell_\infty^{m\times s},
\]
with a context-uniform mixer bound, i.e.,
\(\sup_{\mathbf z\in\ell^q}\|\operator_\infty(\disturbance,\mathbf z)\|_\infty<\infty\)
for every \(\disturbance\in\ell_{\infty,\kappa}^n\),
such that
\begin{equation}\label{eq:env-factor}
\operator(\disturbance,\mathbf z)
=
\operator_\infty(\disturbance,\mathbf z)\boxtimes \operator_p(\disturbance),
\qquad
\forall(\disturbance,\mathbf z).
\end{equation}
\end{enumerate}
\end{theorem}

\begin{proof}
\emph{(i)$\Rightarrow$(ii).}
Assume \eqref{eq:env-iff}. It is enough to exhibit one admissible factorization.
Define the causal map
\begin{equation}\label{eq:vp-canonical}
\operator_p(\disturbance)_t := \kappa_t^{-1}\,\mathbf 1_s,\qquad t\ge 0,
\end{equation}
and define \(\operator_\infty\) by
\begin{equation}\label{eq:Minfty-canonical}
\big(\operator_\infty(\disturbance,\mathbf z)_t\big)^{i,j}
:=
\begin{cases}
\kappa_t\,\operator(\disturbance,\mathbf z)_t^{i}, & j=1,\\
0, & j\neq 1.
\end{cases}
\end{equation}
Then \eqref{eq:Minfty-canonical}--\eqref{eq:vp-canonical} give
\[
\operator_\infty(\disturbance,\mathbf z)_t\,
\operator_p(\disturbance)_t
=
\operator(\disturbance,\mathbf z)_t
\]
for all \(t\), hence \eqref{eq:env-factor} holds.
Moreover, \(\operator_p(\disturbance)\in \ell_{\infty,\kappa}^s\) since
\[
\|\operator_p(\disturbance)\|_{\infty,\kappa}
=
\sup_{t\ge 0}\kappa_t\|\kappa_t^{-1}\mathbf 1_s\|
=
\|\mathbf 1_s\|<\infty.
\]
Finally, we verify the context-uniform bound on
$\operator_\infty$. Let
$e_1=(1,0,\ldots,0)^\top\in\mathbb R^s$. By construction,
\[
\operator_\infty(\disturbance,\mathbf z)_t
=
\kappa_t\,
\operator(\disturbance,\mathbf z)_t e_1^\top.
\]
Define the finite constant
\[
c_s
:=
\sup_{\|v\|=1}|e_1^\top v|.
\]
The induced-norm definition gives
\[
\|\operator_\infty(\disturbance,\mathbf z)_t\|
\le
c_s\,\kappa_t
\|\operator(\disturbance,\mathbf z)_t\|.
\]
For Euclidean vector norms, $c_s=1$.
Taking the supremum over \(t\) and then over \(\mathbf z\) yields
\(
\sup_{\mathbf z\in\ell^q}\|\operator_\infty(\disturbance,\mathbf z)\|_\infty
\le
c_s\,\sup_{\mathbf z\in\ell^q}\|\operator(\disturbance,\mathbf z)\|_{\infty,\kappa}
<\infty
\)
by \eqref{eq:env-iff}.

\smallskip
\emph{(ii)$\Rightarrow$(i).}
Assume \eqref{eq:env-factor}. For each \(t\ge 0\),
{\small
\begin{align*}
\kappa_t\|\operator(\disturbance,\mathbf z)_t\|
 & =
\kappa_t\|\operator_\infty(\disturbance,\mathbf z)_t
\,\operator_p(\disturbance)_t\| \\
& \le
\|\operator_\infty(\disturbance,\mathbf z)\|_\infty\,
\kappa_t\|\operator_p(\disturbance)_t\|.
\end{align*}}
Taking the supremum over \(t\) and then over \(\mathbf z\) gives
\[
\sup_{\mathbf z\in\ell^q}\|\operator(\disturbance,\mathbf z)\|_{\infty,\kappa}
\le
\Big(\sup_{\mathbf z\in\ell^q}\|\operator_\infty(\disturbance,\mathbf z)\|_\infty\Big)
\|\operator_p(\disturbance)\|_{\infty,\kappa}
<\infty,
\]
since \(\operator_p(\disturbance)\in \ell_{\infty,\kappa}^s\) and the mixer bound is uniform in
\(\mathbf z\). This proves \eqref{eq:env-iff}.
\end{proof}

{\color{black}
\begin{remark}[Interpretation and scope]
\label{rem:typical-pb}
The construction in
\eqref{eq:vp-canonical}--\eqref{eq:Minfty-canonical} is an existence
argument rather than a design prescription. More specifically, to establish the existence of the factorization, the proof selects
the fixed disturbance-independent sequence
$\operator_p(\disturbance)_t
:=
\kappa_t^{-1}\mathbf{1}_s.$
This sequence serves only as a mathematical witness and is not intended to
prescribe the practical implementation of $\operator_p$.

In practice, $\operator_p$ and $\operator_\infty$ are designed directly.
Theorem~\ref{thm:suff} guarantees $\mathcal{L}_p$-stability by construction
whenever $\operator_p$ is $\mathcal{L}_p$-stable and
$\operator_\infty$ is uniformly bounded. Theorem~\ref{thm:equiv} shows that the abstract bounded-mixer
factorization is exact, on the considered weighted-envelope domain, for the
class of causal envelope-preserving operators. This conclusion concerns the
operator classes appearing in the theorem and does not imply completeness of
a particular finite-dimensional parametrization of the two factors.
Nevertheless, it provides a theoretical justification for directly optimizing
expressive parametrizations of $\operator_p$ and $\operator_\infty$, rather
than using the canonical construction employed in the proof.

Consequently, if $\disturbance\in\ell_p^n$ but
$\disturbance\notin\ell_{\infty,\kappa}^n$, the sufficient stability
guarantee remains valid, while the converse representability result is
no longer ensured. 
\end{remark}
}

{\begin{remark}[Robustness to model mismatch]\label{rem:mismatch}
The results above assume perfect knowledge of $\mathbf F$ for the disturbance reconstruction. When only an approximate model $\widehat{\mathbf F}\in\mathcal L_p$ is available, the reconstructed signal becomes $\disturbance = \mathbf w + (\mathbf F-\widehat{\mathbf F})(\mathbf x,\mathbf u)$ and the reconstruction is no longer exact. As shown in \cite{10633771} for PB and in \cite{11312988} for rPB, closed-loop $\mathcal L_p$-stability can nevertheless be preserved through a small-gain argument, provided the gain of the boosting operator is sufficiently small with respect to that of the mismatch $\mathbf F-\widehat{\mathbf F}$. The factorization \eqref{eq:op-decomp} is naturally suited to this robustification. Since the mixer is uniformly bounded by construction with a known constant, any prescribed gain for the overall operator can be enforced by detuning the gain of $\operator_p$ alone, which is directly implementable since the SSM parametrizations used for $\operator_p$ admit prescribed $\mathcal L_2$-gain bounds by design \cite{11312999}. Robustness is thus achieved without constraining the context path: the mixer retains its full expressiveness, at the cost of the usual robustness--performance trade-off induced by reducing the corrective gain.
\end{remark}}

\textcolor{black}{
\begin{example}[Necessity fails on $\ell_p$ without envelopes]\label{ex:no-envelope}
To illustrate why the envelope regime cannot be dispensed with, let
$n=m=q=s=1$ and fix any $\disturbance\in\ell_p$. Consider the causal
``first-trigger echo''
\begin{equation*}
\operator(\disturbance,\mathbf z)_t :=
\begin{cases}
1, & z_t\neq 0 \ \text{and}\ z_s=0 \ \ \forall s<t,\\
0, & \text{otherwise},
\end{cases}
\end{equation*}
whose output has at most one nonzero entry; hence
$\operator(\disturbance,\mathbf z)\in\ell_p$ for \emph{every} $\mathbf z\in\ell^q$,
so $\operator$ is an admissible extended operator of the form
\eqref{eq:extended_operator}. For the trigger contexts defined by
$z^{(T)}_t=1$ if $t=T$ and $z^{(T)}_t=0$ otherwise, one has
$\operator(\disturbance,\mathbf z^{(T)})=\mathbf z^{(T)}$.
Suppose now that $\operator$ admitted a factorization \eqref{eq:op-decomp} on
$\ell_p$ with $\operator_p(\disturbance)\in\ell_p$ and a context-uniform mixer
bound $\sup_{\mathbf z}\|\operator_\infty(\disturbance,\mathbf z)\|_\infty=:B<\infty$.
Evaluating the factorization at $t=T$ yields
\begin{equation*}
1 = \operator_\infty(\disturbance,\mathbf z^{(T)})_T\, \operator_p(\disturbance)_T,
\end{equation*}
whence $\operator_p(\disturbance)_T\neq 0$ and
$|\operator_p(\disturbance)_T|\ge 1/B$ for every $T\ge 0$. This contradicts
$\lim_{t\to\infty}\operator_p(\disturbance)_t=0$, which holds for any
$\ell_p$ sequence with $p\in[1,\infty)$. Hence no factorization with a
context-uniform mixer bound exists on $\ell_p$. Note that a factorization with
a merely input-dependent mixer norm does exist --- e.g.,
$\operator_p(\disturbance)_t=2^{-t}$ and
$\operator_\infty(\disturbance,\mathbf z)_t=2^{t}\operator(\disturbance,\mathbf z)_t$,
whose norm $2^T$ is finite for each trigger but unbounded across the family ---
which is why the uniformity in $\mathbf z$ required by
Definition~\ref{ass:env} and Theorem~\ref{thm:equiv} is the essential
ingredient. The envelope regime excludes this pathology precisely because the
echo is not envelope-preserving: its response can be delayed arbitrarily by the
context, so
$\sup_{\mathbf z}\,\sup_{t}\kappa_t\|\operator(\disturbance,\mathbf z)_t\|
=\sup_{T\ge 0}\kappa_T=\infty$, violating \eqref{eq:env}.
\end{example}
}

We next show that common classes of operators used within the PB framework are
envelope-preserving for polynomial envelopes
$\kappa_t=(t+1)^\alpha$. The corresponding admissible class
$\ell_{\infty,\kappa}^n$ accommodates a broad range of practically relevant
$\ell_p$ disturbances, including finite-support transients and signals with
polynomially decaying tails.

\begin{lemma}[Fading memory preserves envelopes]\label{lem:fading-poly}
Fix $\alpha>0$, $\kappa_t=(t+1)^\alpha$, and let
$\mathcal M\colon\ell^n\to\ell^r$ be causal with the exponential forgetting property: 
{\small\begin{equation}\label{eq:fading}
\|\mathcal M(u)_t\|
\le a\rho^t + b\sum_{k=0}^{t}\rho^{t-k}\|u_k\|,
\qquad a,b\ge0,\;\rho\in(0,1).
\end{equation}}
Then $\mathcal M$ maps $\ell_{\infty,\kappa}^n$ into $\ell_{\infty,\kappa}^r$ and
\begin{equation}\label{eq:fading-poly-bound}
\|\mathcal M(u)\|_{\infty,\kappa}
\le a\sup_{t\ge0}(t{+}1)^\alpha\rho^t
  + b\,C_{\alpha,\rho}\,\|u\|_{\infty,\kappa},
\end{equation}
where $C_{\alpha,\rho}:=\sum_{j=0}^\infty\rho^{j}(j{+}1)^\alpha<\infty$.
\end{lemma}

\begin{proof}
For $u\in\ell_{\infty,\kappa}^n$ we have
$\|u_k\|\le\|u\|_{\infty,\kappa}(k+1)^{-\alpha}$, so with $j=t-k$,
\[
(t{+}1)^\alpha\|\mathcal M(u)_t\|
\le a(t{+}1)^\alpha\rho^t
  + b\|u\|_{\infty,\kappa}
    \sum_{j=0}^{t}\rho^j\frac{(t{+}1)^\alpha}{(t{-}j{+}1)^\alpha}.
\]
Since $0\le j\le t$ implies $(j{+}1)(t{-}j{+}1)\ge t{+}1$, the ratio is at most $(j{+}1)^\alpha$.
Since all summands are non-negative, replacing $\sum_{j=0}^{t}$ by
$\sum_{j=0}^{\infty}=C_{\alpha,\rho}$ yields a bound independent of~$t$;
taking $\sup_t$ gives~\eqref{eq:fading-poly-bound}.
The series converges because exponential decay dominates polynomial growth.
\end{proof}
The exponential-forgetting hypothesis of
Lemma~\ref{lem:fading-poly} is satisfied by the stable REN and SSM
parametrizations considered in this work. For RENs, the certified
contraction property yields exponential decay of the internal-state
response and a geometrically weighted bound on the response to past
inputs \cite{revay2023recurrent}. For the considered SSMs, the
underlying finite-dimensional LTI state matrices are Schur stable.
Consequently, for each state matrix $A$, there exist constants
$c>0$ and $\rho\in(0,1)$ such that
\[
\|A^t\|\le c\rho^t,\qquad t\ge0.
\]
Combined with the Lipschitz continuity and zero-at-zero property of
the static nonlinearities, this gives constants $a,b\ge0$ and
$\bar\rho\in(0,1)$ such that
\[
\|\operator_p(\disturbance)_t\|
\le
a\bar\rho^t
+
b\sum_{k=0}^{t}
\bar\rho^{\,t-k}\|\disturbance_k\|.
\]
A finite cascade of such layers preserves this estimate: composing two bounds
of the form \eqref{eq:fading} produces additional terms bounded by
$(t+1)\bar\rho^{\,t}$, which are absorbed into $a'\tilde\rho^{\,t}$ for any
$\tilde\rho\in(\bar\rho,1)$, so the composition again satisfies
\eqref{eq:fading} with adjusted constants and decay rate. Hence the REN and SSM disturbance processors
used in PB map polynomial weighted-envelope spaces into themselves
and satisfy the hypothesis of
Lemma~\ref{lem:fading-poly}.

\section{Numerical Experiment: Continuous Moving-Gate Navigation}
  \label{sec:exp-gate}
Here, we evaluate the proposed context-aware PB architecture against a
context-agnostic PB baseline and the competing context-aware architectures MAD
and rPB. We consider a challenging navigation task in which a planar robot must
reach the origin while crossing a wall through a continuously moving gate. The
controller must use causal gate observations and their recent history to decide
when and where to cross, while the corrective operator retains the structural
$\mathcal L_p$ constraint imposed by the proposed factorization.

\subsection{The setup}
Consider a planar robot described by
\begin{equation}
x_t=
\begin{bmatrix}
p_t^\top & v_t^\top
\end{bmatrix}^\top
\in\mathbb{R}^4,
\qquad
u_t\in\mathbb{R}^2,
\end{equation}
where $p_t,v_t\in\mathbb{R}^2$ are the robot position and velocity, and $u_t$ is the PB corrective input. We denote by $p_{1,t}$ and $p_{2,t}$ the longitudinal and lateral components of $p_t$, respectively. The nominal model is a pre-stabilized double integrator, augmented with dissipative quadratic drag,
\begin{align}
p_{t+1} &= p_t + \Delta t\, v_t, \label{eq:gate-nom-p}\\
v_{t+1} &= v_t + \Delta t\!\left(
-K_p p_t - K_d v_t - c_{\mathrm{drag}}\|v_t\|_2v_t + u_t
\right),
\label{eq:gate-nom-v}
\end{align}
with $\Delta t=0.05$\,s, $K_p=0.32$, $K_d=0.80$, and $c_{\mathrm{drag}}=1.0$. We denote by $f_{\mathrm{nom}}$ the nominal transition map associated with \eqref{eq:gate-nom-p}--\eqref{eq:gate-nom-v}. Notice that the robot is pre-stabilized around the target position at the origin, $x_{\mathrm{goal}}=0$, in the absence of any PB correction.

The true closed-loop dynamics are driven by an additive process disturbance,
\begin{equation}
x_{t+1}=f_{\mathrm{nom}}(x_t,u_t)+w_{t+1},
\qquad t\ge 0.
\label{eq:gate-rollout}
\end{equation}
The initial state is encoded by $x_0=w_0$: the robot is initialized at rest, with
\begin{align}
p_{1,0}\sim\mathcal{U}[0.6,2.1] \: \: \mbox{ and } \: \:
p_{2,0}\sim\mathcal{U}[-1.5,1.5]. \label{eq:init}\end{align}
For $t\ge 1$, the disturbance combines small background noise with transient velocity gusts,
\begin{equation}
w_t=\chi_t\!\left(\eta_t+\sum_{k=1}^{K}b_k\,\mathbf{1}_{[t_k,t_k+\delta_k)}(t)\right),
\label{eq:gate-noise}
\end{equation}
where $\eta_t$ is i.i.d.\ Gaussian with standard deviations
$\sigma_{\mathrm{pos}}=3\times10^{-4}$ on position channels and
$\sigma_{\mathrm{vel}}=1.2\times10^{-3}$ on velocity channels. Each burst affects only the velocity channels, with $K\sim\mathcal{U}\{2,3,4\}$ and
$\delta_k\sim\mathcal{U}\{4,\ldots,10\}$. The longitudinal burst amplitude is sampled from $[-0.004,0.004]$, while the lateral gust magnitude is sampled from $[0.01,0.028]$ with random sign. The factor $\chi_t\in[0,1]$ is equal to one over the early part of the episode and smoothly decays to zero at the horizon, so the disturbance sequence can be extended by zero for $t\ge T$, with $T=160$. Hence the experiment remains within the admissible disturbance class of Theorem~\ref{thm:equiv}.

Following the PB formulation, the disturbance is reconstructed online from one-step nominal prediction errors,
\begin{equation}
\widehat w_0=x_0,
\qquad
\widehat w_t=x_t-f_{\mathrm{nom}}(x_{t-1},u_{t-1}),
\quad t\ge 1.
\label{eq:gate-wt}
\end{equation}
Under \eqref{eq:gate-rollout}, this reconstruction is exact, so $\widehat w_t=w_t$ at every step.

The robot moves inside a corridor of half-width $y_s=1.6$. A vertical wall is placed at longitudinal position $p_{\mathrm{wall}}=0.55$ and contains a moving gate of half-width $h=0.20$. The gate centre $g_t\in[-A,A]$, with $A=0.95$, evolves as a clipped Ornstein--Uhlenbeck process over the whole episode:
\begin{align}
\mu &\sim \mathcal{U}[-\rho_g A,\rho_g A], \notag\\
g_{t+1}
&=
\Pi_{[-A,A]}\!\left(
\mu+(1-\theta)(g_t-\mu)+\sigma_g\xi_{t+1}
\right),
\quad
\xi_t\sim\mathcal{N}(0,1),
\label{eq:continuous-gate}
\end{align}
where $\Pi_{[-A,A]}(r)=\min\{A,\max\{-A,r\}\}$ clips the gate within the corridor, $\theta=1-\exp(-1/\tau_g)$ and
$\sigma_g=r_g\sqrt{2\theta-\theta^2}$. In the reported runs we use
$\tau_g=60$, $r_g=0.50$, and $\rho_g=0.55$. The gate therefore has a smooth but persistent motion around an episode-dependent centre $\mu$ and continues to evolve while the robot approaches the wall. The overall setup is illustrated in Fig.~\ref{fig:gate-setup-continuous}.

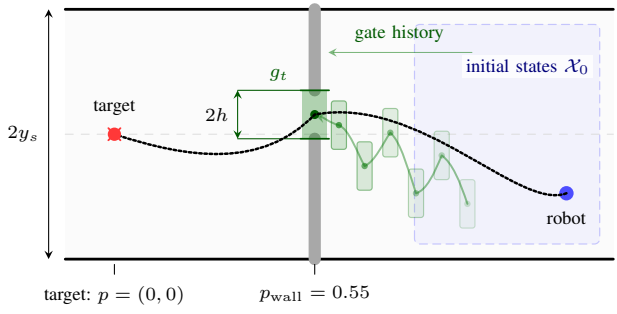
\begin{figure}[t]
  \centering
\begin{tikzpicture}[
  x=1.00cm,
  y=1.00cm,
  >=stealth,
  line cap=round,
  line join=round,
  font=\scriptsize
]
  \def\xmin{0.20}
  \def\xmax{7.45}
  \def\ytop{1.65}
  \def\ybot{-1.65}
  \def\xwall{3.50}
  \def\xrobot{6.83}
  \def\yrobot{-0.78}
  \def\xtarget{0.85}
  \def\ytarget{0.00}
  \def\xinitmin{4.82}
  \def\xinitmax{7.27}
  \def\yinitmin{-1.45}
  \def\yinitmax{1.45}
  \def\hgate{0.32}
  \def\gA{-0.92}
  \def\gB{-0.28}
  \def\gC{-0.78}
  \def\gD{0.02}
  \def\gE{-0.42}
  \def\gF{0.12}
  \def\gnow{0.26}

  \fill[gray!3] (\xmin,\ybot) rectangle (\xmax,\ytop);
  \filldraw[
    fill=blue!12,
    fill opacity=0.30,
    draw=blue!45,
    draw opacity=0.48,
    densely dashed,
    rounded corners=2pt,
    line width=0.45pt
  ] (\xinitmin,\yinitmin) rectangle (\xinitmax,\yinitmax);

  \draw[line width=1.05pt]
    (\xmin,\ytop) -- (\xmax,\ytop);
  \draw[line width=1.05pt]
    (\xmin,\ybot) -- (\xmax,\ybot);
  \draw[gray!26, thin, dashed]
    (\xmin,0) -- (\xmax,0);
  \draw[<->]
    (\xmin-0.22,\ybot) -- (\xmin-0.22,\ytop)
    node[midway,left] {$2y_s$};

  \draw[line width=4.6pt, gray!68]
    (\xwall,\ytop) -- (\xwall,\gnow+\hgate);
  \draw[line width=4.6pt, gray!68]
    (\xwall,\gnow-\hgate) -- (\xwall,\ybot);

  \fill[green!45!black, opacity=0.045, rounded corners=1pt]
    (\xwall+1.92,\gA-\hgate)
    rectangle
    (\xwall+2.12,\gA+\hgate);
  \fill[green!45!black, opacity=0.060, rounded corners=1pt]
    (\xwall+1.58,\gB-\hgate)
    rectangle
    (\xwall+1.78,\gB+\hgate);
  \fill[green!45!black, opacity=0.080, rounded corners=1pt]
    (\xwall+1.24,\gC-\hgate)
    rectangle
    (\xwall+1.44,\gC+\hgate);
  \fill[green!45!black, opacity=0.105, rounded corners=1pt]
    (\xwall+0.90,\gD-\hgate)
    rectangle
    (\xwall+1.10,\gD+\hgate);
  \fill[green!45!black, opacity=0.135, rounded corners=1pt]
    (\xwall+0.56,\gE-\hgate)
    rectangle
    (\xwall+0.76,\gE+\hgate);
  \fill[green!45!black, opacity=0.170, rounded corners=1pt]
    (\xwall+0.22,\gF-\hgate)
    rectangle
    (\xwall+0.42,\gF+\hgate);

  \draw[
    green!45!black,
    line width=0.45pt,
    opacity=0.18,
    rounded corners=1pt
  ]
    (\xwall+1.92,\gA-\hgate)
    rectangle
    (\xwall+2.12,\gA+\hgate);
  \draw[
    green!45!black,
    line width=0.45pt,
    opacity=0.23,
    rounded corners=1pt
  ]
    (\xwall+1.58,\gB-\hgate)
    rectangle
    (\xwall+1.78,\gB+\hgate);
  \draw[
    green!45!black,
    line width=0.45pt,
    opacity=0.28,
    rounded corners=1pt
  ]
    (\xwall+1.24,\gC-\hgate)
    rectangle
    (\xwall+1.44,\gC+\hgate);
  \draw[
    green!45!black,
    line width=0.45pt,
    opacity=0.34,
    rounded corners=1pt
  ]
    (\xwall+0.90,\gD-\hgate)
    rectangle
    (\xwall+1.10,\gD+\hgate);
  \draw[
    green!45!black,
    line width=0.45pt,
    opacity=0.40,
    rounded corners=1pt
  ]
    (\xwall+0.56,\gE-\hgate)
    rectangle
    (\xwall+0.76,\gE+\hgate);
  \draw[
    green!45!black,
    line width=0.45pt,
    opacity=0.46,
    rounded corners=1pt
  ]
    (\xwall+0.22,\gF-\hgate)
    rectangle
    (\xwall+0.42,\gF+\hgate);

  \fill[green!45!black, opacity=0.30]
    (\xwall-0.16,\gnow-\hgate)
    rectangle
    (\xwall+0.16,\gnow+\hgate);
  \draw[green!35!black, line width=0.65pt]
    (\xwall-0.19,\gnow+\hgate)
    --
    (\xwall+0.19,\gnow+\hgate);
  \draw[green!35!black, line width=0.65pt]
    (\xwall-0.19,\gnow-\hgate)
    --
    (\xwall+0.19,\gnow-\hgate);

  \draw[
    green!45!black,
    line width=0.75pt,
    opacity=0.42,
    ->
  ]
    (\xwall+2.02,\gA)
    .. controls
      (\xwall+1.92,-0.62) and
      (\xwall+1.82,-0.38) ..
    (\xwall+1.68,\gB)
    .. controls
      (\xwall+1.58,-0.52) and
      (\xwall+1.48,-0.74) ..
    (\xwall+1.34,\gC)
    .. controls
      (\xwall+1.24,-0.48) and
      (\xwall+1.14,-0.10) ..
    (\xwall+1.00,\gD)
    .. controls
      (\xwall+0.90,-0.12) and
      (\xwall+0.80,-0.36) ..
    (\xwall+0.66,\gE)
    .. controls
      (\xwall+0.56,-0.24) and
      (\xwall+0.46,0.06) ..
    (\xwall+0.32,\gF)
    .. controls
      (\xwall+0.20,0.15) and
      (\xwall+0.12,0.22) ..
    (\xwall,\gnow);

  \fill[green!45!black, opacity=0.16]
    (\xwall+2.02,\gA) circle (1.2pt);
  \fill[green!45!black, opacity=0.24]
    (\xwall+1.68,\gB) circle (1.2pt);
  \fill[green!45!black, opacity=0.32]
    (\xwall+1.34,\gC) circle (1.2pt);
  \fill[green!45!black, opacity=0.42]
    (\xwall+1.00,\gD) circle (1.2pt);
  \fill[green!45!black, opacity=0.52]
    (\xwall+0.66,\gE) circle (1.2pt);
  \fill[green!45!black, opacity=0.62]
    (\xwall+0.32,\gF) circle (1.2pt);
  \fill[green!35!black]
    (\xwall,\gnow) circle (1.7pt);

  \draw[green!45!black, ->, opacity=0.50]
    (\xwall+2.08,1.08) -- (\xwall+0.18,1.08);
  \node[
    green!35!black,
    anchor=south
  ] at (\xwall+1.12,1.10)
    {gate history};
  \node[
    green!35!black,
    anchor=east
  ] at (\xwall-0.22,\gnow+0.52)
    {$g_t$};

  \draw[green!35!black, thin]
    (\xwall-0.20,\gnow+\hgate)
    --
    (\xwall-1.02,\gnow+\hgate);
  \draw[green!35!black, thin]
    (\xwall-0.20,\gnow-\hgate)
    --
    (\xwall-1.02,\gnow-\hgate);
  \draw[<->, green!35!black]
    (\xwall-1.02,\gnow-\hgate)
    --
    (\xwall-1.02,\gnow+\hgate)
    node[midway,left,black] {$2h$};

  \fill[blue!70]
    (\xrobot,\yrobot) circle (2.7pt);
  \node[anchor=north]
    at (\xrobot,\yrobot-0.12)
    {robot};
  \node[
    blue!48!black,
    anchor=north east,
    align=right,
    fill=white,
    fill opacity=0.78,
    text opacity=1,
    inner sep=1.4pt
  ] at (\xinitmax-0.08,\yinitmax-0.40)
    {initial states $\mathcal X_0$};

  \draw[densely dotted, line width=0.95pt]
    (\xrobot,\yrobot)
    .. controls
      (6.20,-1.02) and
      (4.90,0.52) ..
    (\xwall,\gnow)
    .. controls
      (2.80,-0.46) and
      (1.83,-0.32) ..
    (\xtarget,\ytarget);

  \fill[red!78]
    (\xtarget,\ytarget) circle (2.5pt);
  \draw[red!78, line width=0.7pt]
    (\xtarget-0.08,\ytarget-0.08)
    --
    (\xtarget+0.08,\ytarget+0.08)
    (\xtarget-0.08,\ytarget+0.08)
    --
    (\xtarget+0.08,\ytarget-0.08);
  \node[anchor=south]
    at (\xtarget,\ytarget+0.13)
    {target};

  \draw[thin]
    (\xtarget,\ybot-0.07)
    --
    (\xtarget,\ybot-0.20);
  \node[anchor=north]
    at (\xtarget,\ybot-0.23)
    {target: $p=(0,0)$};

  \draw[thin]
    (\xwall,\ybot-0.07)
    --
    (\xwall,\ybot-0.20);
  \node[anchor=north]
    at (\xwall,\ybot-0.23)
    {$p_{\mathrm{wall}}=0.55$};
\end{tikzpicture}
  \caption{Continuous moving-gate navigation setup. The wall opening drifts throughout the episode, so the controller must cross the wall using only causal information about the current gate motion.}
  \label{fig:gate-setup-continuous}
\end{figure}

The goal of the PB corrective input is to steer the robot to the origin while avoiding collisions with both the corridor boundaries and the moving wall. To do that, for each sample trajectory, we optimize the following objective:
{\small
\begin{align}
\mathcal{J}
&=
\underbrace{\lambda_T \|p_T\|_2^2
+ \lambda_v \|v_T\|_2^2
+ \frac{\lambda_S}{T}\sum_{t=0}^{T-1}\ell_H(p_t)}_{\text{terminal and stage target cost}}
\notag\\
&\quad+
\underbrace{\sum_{t=0}^{T-1} \beta_t
\Bigl[
\lambda_{\mathrm{tr}}(p_{2,t}-g_t)^2
+
\lambda_{\mathrm{coll}}\,
\phi_{d}\bigl(|p_{2,t}-g_t|-h_{\mathrm{train}}\bigr)
\Bigr]}_{\text{moving-gate tracking and collision avoidance}}
\notag\\
&\quad+
\underbrace{\frac{\lambda_{\mathrm{ctrl}}}{T}\sum_{t=0}^{T-1}\|u_t\|_2^2
+ \frac{\lambda_{\mathrm{corr}}}{T}\sum_{t=0}^{T-1}
\phi_{d'}\bigl(|p_{2,t}|-y_s\bigr)}_{\text{control energy and corridor collision avoidance}},
\label{eq:gate-loss}
\end{align}}
where $\ell_H$ is the Huber loss with $\delta=0.5$,
$\phi_d(z)=d^{-1}\log(1+e^{d z})$ is a scaled softplus with sharpness $d$, and
\[
\beta_t=
\frac{\exp\!\bigl(-\tfrac12(p_{1,t}-p_{\mathrm{wall}})^2/\sigma_{\mathrm{wall}}^2\bigr)}
{\sum_{j=0}^{T-1}\exp\!\bigl(-\tfrac12(p_{1,j}-p_{\mathrm{wall}})^2/\sigma_{\mathrm{wall}}^2\bigr)}
\]
are time-normalized Gaussian weights centered at the wall, with
$\sigma_{\mathrm{wall}}=0.14$. For continuous-gate training, we use a slightly conservative effective half-width $h_{\mathrm{train}}=h-\Delta h$ with $\Delta h=0.04$, while all reported collision metrics use the true half-width $h$. All remaining cost weights and sharpness parameters are listed in the accompanying code. It is worth noting that the task loss is highly nonlinear and non-convex, and it is coupled with the nonlinear dynamics \eqref{eq:gate-nom-p}--\eqref{eq:gate-nom-v}, making the resulting optimal control problem \eqref{eq:opt_pb} very challenging to solve. This difficulty is compounded by the requirement that the optimized controller preserve closed-loop stability. As anticipated earlier, the PB framework addresses this issue structurally: stability is enforced by the controller architecture rather than by the task objective, allowing the stable operator to be optimized for the navigation loss without compromising asymptotic stability.

\subsection{Context, PB controller design and training procedure}
\label{subsec:gate-context}

At each time step, the context-aware controller receives a causal context signal $z_t$ that summarizes the available information about the moving gate and the robot's current geometric relation to it. We stress that this representation of the contextual signal is not fixed by the framework: the specific choice of $z_t$ is a design degree of freedom, provided that its components remain bounded. In our experiments, we use the following contextual signal:
\begin{equation}
\small
\begin{split}
z_t =
\Big[&
g_t/y_s,\;
\Delta g_t/y_s,\;
\bar g_t/y_s,\;
(p_{2,t}-g_t)/y_s,\;
(p_{1,t}-p_{\mathrm{wall}})/x_s,\\
&
-p_{1,t}/x_s,\;
-p_{2,t}/y_s,\;
v_{1,t},\;
v_{2,t}
\Big]^{\!\top}.
\end{split}
\label{eq:gate-context}
\end{equation}
Here, $\Delta g_t=g_t-g_{t-1}$ is the causally observed gate velocity, while $\bar g_t$ is an exponential moving average of the gate centre,
\begin{align} \label{eq:gt}
    \bar g_t=\alpha_g g_t+(1-\alpha_g)\bar g_{t-1},
\qquad
\bar g_0=g_0,
\end{align}
with $\alpha_g=0.35$.
The constants $y_s$ (the corridor half-width) and $x_s=2.1$ (the largest
initial longitudinal distance) act as fixed normalization scales for the
lateral and longitudinal features, respectively.

Let us break down the information contained in $z_t$. The first three features describe the gate itself: its current position, instantaneous motion, and recent average position. The term $(p_{2,t}-g_t)/y_s$ gives the signed lateral error between the robot and the gate opening. The feature $(p_{1,t}-p_{\mathrm{wall}})/x_s$ indicates how close the robot is to the wall, while $-p_{1,t}/x_s$ and $-p_{2,t}/y_s$ encode the normalized displacement from the goal at the origin. Finally, $v_{1,t}$ and $v_{2,t}$ provide the current velocity, allowing the controller to adapt its correction to the robot's dynamical state.
The PB correction is computed using the factorized context-enriched operator introduced in \eqref{eq:op-decomp},
\[
\operator(\disturbance,\mathbf z)
=
\operator_\infty(\disturbance,\mathbf z)
\boxtimes
\operator_p(\disturbance).
\]
For this experiment, \(n=4\), \(q=9\), and \(m=2\). We select a factorization dimension \(s=16\), so that
\[
\operator_p:\ell_p^4\rightarrow\ell_p^{16},
\qquad
\operator_\infty:\ell_p^4\times\ell^9
\rightarrow\ell_\infty^{2\times16}.
\]

The disturbance-processing operator $\operator_p$ is implemented as a deep
stable structured state-space model (SSM) following the parametrization of
\cite{11312999}. We use $N_p=8$ SSM layers, hidden dimension $d_p=20$, GLU
feedforward nonlinearities, and output dimension $s=16$. The bounded mixer
$\operator_\infty$ is implemented as an MLP applied pointwise to
$(\widehat w_t,z_t)$, with depth $N_\infty=4$, hidden width $d_\infty=64$, and
output dimension $ms=32$, reshaped into a $2\times16$ matrix. Its linear
layers are spectrally normalized, and its output is passed through an
elementwise scaled softsign map satisfying
\[
\left|
\operator_\infty(\disturbance,\mathbf z)_t^{i,j}
\right|
\leq B_\infty,
\qquad
B_\infty=8.
\]
Consequently, under the Euclidean induced norm,
\[
\left\|
\operator_\infty(\disturbance,\mathbf z)_t
\right\|_2
\leq
\left\|
\operator_\infty(\disturbance,\mathbf z)_t
\right\|_{\mathrm F}
\leq
B_\infty\sqrt{ms}
=
8\sqrt{32}.
\]
Thus $\operator_\infty(\disturbance,\mathbf z)$ is uniformly bounded, while
$\operator_p(\disturbance)\in\ell_p^{16}$, as required by
Theorem~\ref{thm:suff}.

For the context-agnostic PB baseline, we retain the same factorized architecture but set \(z_t=0\) at every time step. The mixer may therefore depend on the reconstructed disturbance, but it cannot condition the corrective action on the realized gate trajectory. We also compare against parameter-matched implementations of MAD \cite{furieri2025mad} and rPB \cite{11312988}, obtained as special cases of \eqref{eq:op-decomp}. As mentioned in section \ref{sec:main}, MAD uses a scalar disturbance component, corresponding to \(s=1\), whereas rPB restricts \(\operator_\infty(\disturbance,\mathbf z)_t\) to be diagonal, with \(s=m=2\). Both baselines receive the complete context signal in \eqref{eq:gate-context}; their SSM hyperparameters are so that the overall numbers of trainable parameters match that of the proposed architecture as closely as possible. Finally, the context-ablation variants are obtained by providing \(\operator_\infty\) with selected subsets of the features in \eqref{eq:gate-context}.

We train the PB controllers by solving \eqref{eq:opt_pb} as described in Sec.~\ref{sec:pb}. At each epoch, we generate a fresh batch of $4096$ trajectories of length $T=160$, with independently sampled initial conditions, continuous gate trajectories, and disturbance sequences according to \eqref{eq:init}, \eqref{eq:gate-noise} and \eqref{eq:continuous-gate} respectively. The samples are generated in paired form: each base scenario is accompanied by
a gate-reflected companion, obtained using the same initial condition and
disturbance sequence while replacing $g_t$ by $-g_t$. The resulting batch is
then shuffled. Thus, each training epoch contains $2048$ independent base
scenarios and $2048$ corresponding gate-reflected companions. We train controllers for $800$ epochs. Model selection is performed on a fixed validation batch of $4096$ held-out trajectories, and final performance is reported on an independent fixed test batch of $N=4096$ trajectories. For evaluation, we use held-out trajectories sampled from the same continuous-gate distribution and report:
\begin{itemize}
\item \textbf{Success rate}: fraction of trajectories that both cross the
moving gate without collision and terminate within
$\varepsilon_{\mathrm{goal}}=0.18$ of the origin.

\item \textbf{Crash rate}: fraction of trajectories that collide with the
wall instead of passing through the moving gate.

\item \textbf{Goal rate}: fraction of trajectories whose terminal position lies within $\varepsilon_{\mathrm{goal}}=0.18$ of the goal,
\[
\left|p_T-p_{\mathrm{goal}}\right|*2
\leq \varepsilon*{\mathrm{goal}},
\]
irrespective of whether the gate-crossing condition is satisfied.

\item \textbf{Average crossing error}: empirical mean of
$\lvert p_{2,t^\star}-g_{t^\star}\rvert$, where $t^\star$ is obtained by
interpolating the trajectory at $p_1=p_{\mathrm{wall}}$.

\item \textbf{Average control energy}: mean quadratic control effort per
time step and trajectory,
\[
E_u=
\frac{1}{NT}
\sum_{i=1}^{N}\sum_{t=0}^{T-1}
\left\|u_t^{(i)}\right\|_2^2.
\]
This metric measures how much control action is required, independently
of the task-performance penalties.

\item \textbf{Average cost}: empirical mean of the complete task objective
evaluated over the test trajectories,
\[
\overline{J}
=
\frac{1}{N}\sum_{i=1}^{N}J^{(i)},
\]
where $J$ is defined in~\eqref{eq:gate-loss}.
\end{itemize}

\subsection{Results}

Here, we present the results obtained with the proposed contextual factorization through two sets of experiments. First, we compare C.~Factorization with the context-agnostic PB controller, MAD, and rPB. Second, we perform a context ablation to assess the contribution of the different contextual features.
\label{subsec:gate-results}

\paragraph{Architecture comparison}

\begin{table*}[t]
  \centering
 \caption{Continuous moving-gate results. Panel~(a) compares the controller
architectures, while Panel~(b) studies the context features supplied to
C.~Factorization. Each panel reports performance over $4096$ test episodes.
Success requires both collision-free gate passage and terminal goal attainment.
Arrows indicate the preferred direction, and the best values are highlighted
in bold within each panel.}
  \label{tab:continuous-gate-comparisons}

  \small
  \setlength{\tabcolsep}{5pt}
  \renewcommand{\arraystretch}{1.18}

  \resizebox{\textwidth}{!}{%
  \begin{tabular}{lcccccc}
    \toprule
    Configuration
    & Success (\%) $\uparrow$
    & Crash (\%) $\downarrow$
    & Goal (\%) $\uparrow$
    & Cross.\ error $\downarrow$
    & Control energy $\downarrow$
    & Cost $\downarrow$ \\
    \midrule

    \rowcolor{sectiongray}
    \multicolumn{7}{l}{\textbf{(a) Architecture comparison}} \\

    \rowcolor{tablegray}
    Context-agnostic
    & 28.20
    & 71.80
    & \textbf{100.00}
    & 0.4074
    & \textbf{1.090}
    & 54.278 \\

    \rowcolor{white}
    MAD
    & 78.64
    & 21.26
    & 99.80
    & 0.1292
    & 1.482
    & 15.101 \\

    \rowcolor{tablegray}
    rPB
    & 74.49
    & 25.51
    & \textbf{100.00}
    & 0.1392
    & 1.717
    & 14.933 \\

    \rowcolor{white}
    C.~Factorization
    & \textbf{83.45}
    & \textbf{16.53}
    & 99.98
    & \textbf{0.1168}
    & 1.493
    & \textbf{12.230} \\

    \midrule

    \rowcolor{sectiongray}
    \multicolumn{7}{l}{\textbf{(b) Context-feature comparison for C.~Factorization}} \\

    \rowcolor{white}
     $z_t^{(0)}$ (no gate info)
& $28.93$ & $70.07$ & $99.98$ & $0.4096$  & \textbf{1.210} & $51.036$ \\

    \rowcolor{tablegray}
    $z_t^{(1)}$  (minimal gate info)
    & 79.39
    & 20.58
    & {99.97}
    & 0.1265
    & {1.449}
    & 14.257 \\

    \rowcolor{white}
     $z_t^{(2)}$  (intermediate gate info)
    & 80.81
    & 19.17
    & 99.97
    & 0.1238
    & 1.504
    & 13.226 \\

    \rowcolor{tablegray}
     $z_t^{(3)}$  (full gate info)
    & \textbf{83.45}
    & \textbf{16.53}
    & \textbf{99.98}
    & \textbf{0.1168}
    & 1.493
    & \textbf{12.230} \\
    \bottomrule
  \end{tabular}%
  }
\end{table*}

Panel~(a) of Table~\ref{tab:continuous-gate-comparisons} compares
C.~Factorization with the context-agnostic PB controller, MAD, and rPB, while
Panel~(b) reports the context-feature ablation. Not surprisingly, goal success is approximately $100\%$ for every
architecture, indicating that the main source of failure is not convergence
to the origin after crossing, but collision with the moving wall. The
context-agnostic PB controller succeeds in only $28.20\%$ of the test episodes
and has an average crossing error of $0.4074$. Although it can react to
reconstructed disturbances, it cannot condition its correction on the
realized gate motion and must therefore implement a context-free compromise
across all possible gate trajectories.

MAD and rPB substantially improve the success rate to $78.64\%$ and
$74.49\%$, respectively, confirming the importance of incorporating contextual
information. In this experiment, MAD performs better than rPB despite its
simpler factorization. C.~Factorization achieves the best overall result, with
an $83.45\%$ success rate, a $16.53\%$ crash rate, and the smallest average
crossing error, $0.1168$. It also attains the lowest average cost, $12.230$. Its average control energy,
$1.493$, is close to that of MAD, $1.482$, indicating that the performance
improvement is not obtained simply by applying substantially larger control
inputs.

Figure~\ref{fig:gate-summary} provides a qualitative explanation for the
aggregate differences in Table~\ref{tab:continuous-gate-comparisons}. We picked a particularly challenging test realization where
the opening changes rapidly when the robot is already approaching the
wall. The context-agnostic controller cannot revise its action according to
this motion, while the corrections produced by MAD and rPB remain
insufficient to align their trajectories with the opening at their respective
wall events. C.~Factorization instead is able to react more effectively to the erratic gate motion, inducing a complex maneuver that allows to cross safely before
returning to the origin. More explanatory animated GIFs, as well as other tests, are available at \url{https://github.com/DecodEPFL/Performance_Boosting}.

\begin{figure*}[t]
  \centering
  \includegraphics[width=\textwidth]{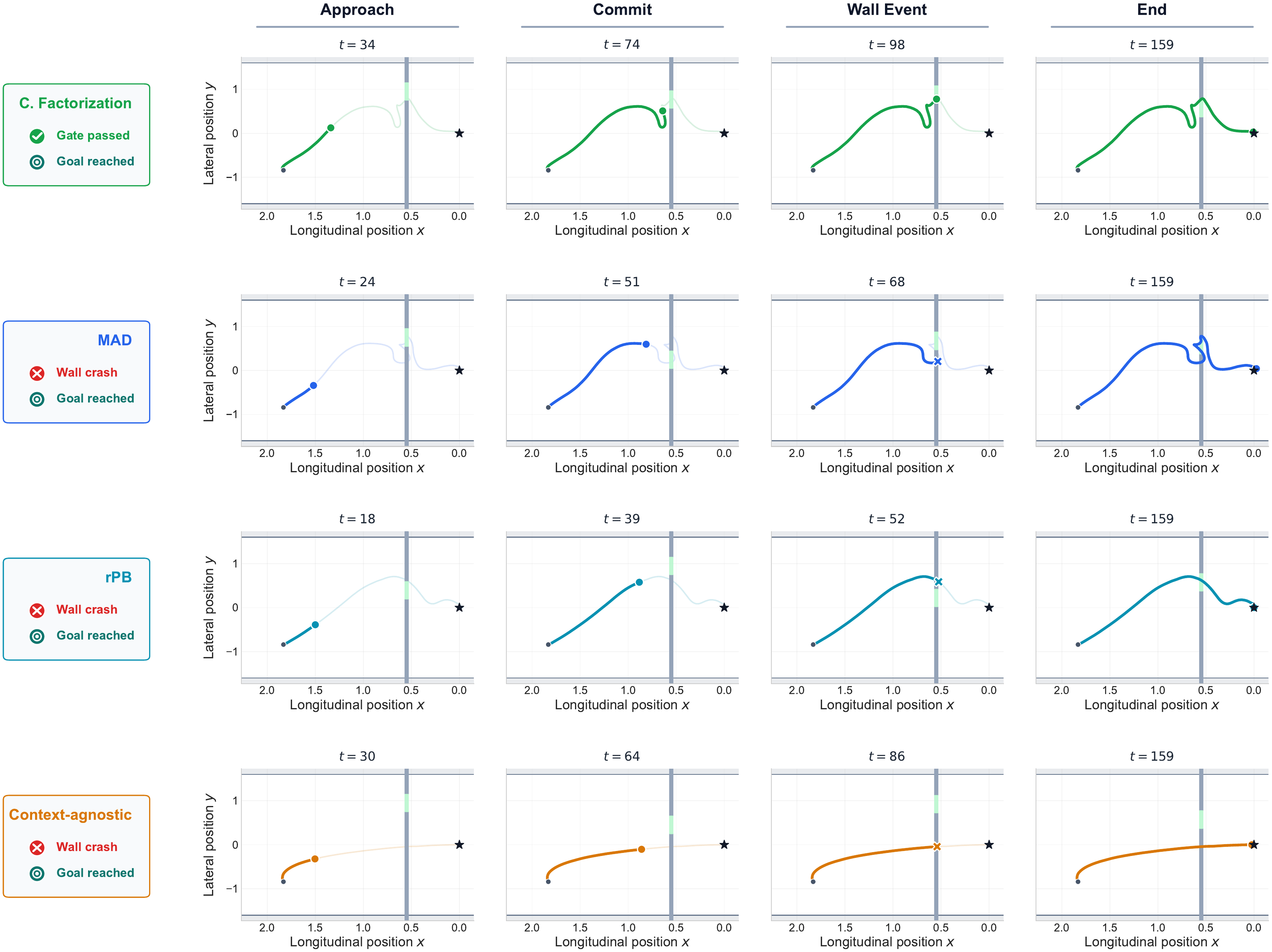}
\caption{Architecture comparison on a particularly challenging continuous-gate realization. All controllers share the same initial condition and gate trajectory. The gate undergoes several rapid displacements during the final approach, including changes immediately before or at the wall-crossing times of rPB ($t=52$), MAD ($t=68$), and the context-agnostic controller ($t=86$). C.~Factorization adapts to this late motion and crosses safely at $t=98$, whereas the other three controllers collide with the wall. Columns show different snapshots, namely, the approach, commitment, controller-specific wall event (crossing or crashing), and final state at the end of the control horizon; bold segments indicate the trajectory completed at the displayed time, while faded segments show its continuation.}
  \label{fig:gate-summary}
\end{figure*}

\paragraph{Context ablation}
The context vector $\mathbf z_t$ in~\eqref{eq:gate-context} combines the
instantaneous gate geometry, recent gate motion, relative goal position, and
agent velocity. To determine which parts of this information are responsible
for the improvement of C.~Factorization, we first introduce a gate-blind
context and then consider three progressively richer gate-aware context sets.
In particular, we define
\begin{equation}
\small
\begin{aligned}
z_t^{(0)}
&=
\Big[
(p_{1,t}-p_{\mathrm{wall}})/x_s,\quad
(p_{1,\mathrm{goal}}-p_{1,t})/x_s,
\\[-0.5mm]
&\qquad
(p_{2,\mathrm{goal}}-p_{2,t})/y_s,\quad
v_{1,t},\quad
v_{2,t}
\Big]^\top,
\\[1mm]
z_t^{(1)}
&=
\Big[
g_t/y_s,\quad
(p_{2,t}-g_t)/y_s,\quad
(p_{1,t}-p_{\mathrm{wall}})/x_s
\Big]^\top,
\\[1mm]
z_t^{(2)}
&=
\Big[
(z_t^{(1)})^\top,\quad
\bar g_t/y_s,\quad
v_{1,t},\quad
v_{2,t}
\Big]^\top,
\\[1mm]
z_t^{(3)}
&=
z_t
\quad\text{as defined in~\eqref{eq:gate-context}.}
\end{aligned}
\label{eq:ctx_abl}
\end{equation}

The gate-blind context $z_t^{(0)}$ contains the agent position relative to
the wall and the target, together with its velocity, but excludes all
information about the position or motion of the gate. It therefore provides
the controller with sufficient information to approach the goal, but not to
identify where or when the moving opening should be crossed. The minimal
gate-aware context $z_t^{(1)}$ contains only the information required to
describe the instantaneous wall geometry: the observed gate position, the
lateral error relative to the opening, and the remaining longitudinal distance
to the wall. The intermediate context $z_t^{(2)}$ additionally provides a
filtered estimate of the gate position and the current agent velocities,
thereby introducing short-term temporal information and the direction of
motion. Finally, $z_t^{(3)}$ is the complete reference context given
by~\eqref{eq:gate-context}.

Panel~(b) of Table~\ref{tab:continuous-gate-comparisons} shows, as expected,
that information about the moving gate is decisive. In the absence of any
gate-dependent information, the success rate is only $28.93\%$, while the
collision rate reaches $70.07\%$, basically matching the poor performance of the context-agnostic PB controller. Providing the minimal gate-aware context
$z_t^{(1)}$ increases the success rate to $79.39\%$, an improvement of
$52.47$ percentage points, and reduces the collision rate to $20.58\%$.
The gate-crossing error similarly decreases from $0.4096$ to $0.1265$, while
the total cost falls from $51.036$ to $14.257$. In contrast, the gate-blind
controller still reaches
the goal in all test rollouts. Not surprisingly, this confirms that its poor success rate is primarily caused by its inability to locate the moving opening, rather than by an inability to drive the agent toward the origin.

Performance improves further as the gate-aware context becomes richer. The
success rate increases from $79.39\%$ with $z_t^{(1)}$ to $80.81\%$ with
$z_t^{(2)}$, and reaches $83.45\%$ with the complete context $z_t^{(3)}$.
Correspondingly, the crash rate decreases from $20.58\%$ to $19.17\%$ and
finally to $16.53\%$. The complete context also yields the smallest crossing
error, $0.1168$ and the lowest
average cost, $12.230$. Notably, C.~Factorization outperforms both MAD and rPB
even when provided only with the minimal gate-aware context $z_t^{(1)}$,
suggesting an architectural advantage beyond the choice of context features.
The improvement from $z_t^{(2)}$ to $z_t^{(3)}$ indicates that the additional
full-context features, particularly the causally observed gate velocity and
the relative target position, provide useful information for adapting the
control action. This improvement requires only a modest increase in control
energy, from $1.449$ with $z_t^{(1)}$ to $1.493$ with $z_t^{(3)}$.

\section{Conclusions}
We introduced a structured factorization for context-enriched Performance
Boosting operators, combining a disturbance-driven $\mathcal L_p$-stable
dynamical module with a uniformly bounded context-dependent mixer. The
resulting architecture provides a modular and interpretable mechanism for
injecting general contextual information into multi-input PB controllers while
retaining the nominal PB stability guarantees. We proved that the factorization
is always sufficient for $\mathcal L_p$-stability and that, on a
weighted-envelope disturbance domain, it is exact for causal operators
satisfying a context-uniform envelope-preservation property. We also showed
that a broad family of fading-memory operators satisfies this regularity
condition. The continuous moving-gate experiment demonstrated that contextual
information can substantially improve collision avoidance and overall control
performance relative to context-agnostic PB, MAD, and rPB. Future work will
investigate richer contextual architectures, less restrictive admissible
domains, and applications to other learning-based control problems.

\bibliographystyle{ieeetr}
\bibliography{bib}

@inproceedings{furieri2025mad,
  title={ {MAD}: A magnitude and direction policy parametrization for stability constrained reinforcement learning},
  author={Furieri, Luca and Shenoy, Sucheth and Saccani, Danilo and Martin, Andrea and Ferrari-Trecate, Giancarlo},
  booktitle={2025 IEEE 64th Conference on Decision and Control (CDC)},
  pages={942--947},
  year={2025},
}

@ARTICLE{10680398,
  author={Stefanini, Elisa and Palmieri, Luigi and Rudenko, Andrey and Hielscher, Till and Linder, Timm and Pallottino, Lucia},
  journal={IEEE Robotics and Automation Letters}, 
  title={Efficient Context-Aware Model Predictive Control for Human-Aware Navigation}, 
  year={2024},
  volume={9},
  number={11},
  pages={9494-9501},
  doi={10.1109/LRA.2024.3461552}}

@INPROCEEDINGS{10920056,
  author={Coppola, Angelo and Di Pace, Roberta and Storani, Facundo and de Luca, Stefano and Santini, Stefania},
  booktitle={2024 IEEE 27th International Conference on Intelligent Transportation Systems (ITSC)}, 
  title={Context-aware Nonlinear {MPC} for Automated Vehicles Embedding Newell's Car-Following Model}, 
  year={2024},
  volume={},
  number={},
  pages={889-894},
  doi={10.1109/ITSC58415.2024.10920056}}

@article{huang2020learning,
  title={Learning-based switched reliable control of cyber-physical systems with intermittent communication faults},
  author={Huang, Xin and Dong, Jiuxiang},
  journal={IEEE/CAA Journal of Automatica Sinica},
  volume={7},
  number={3},
  pages={711--724},
  year={2020},
  publisher={IEEE}
}

@inproceedings{goel2023semantically,
  title={Semantically informed {MPC} for context-aware robot exploration},
  author={Goel, Yash and Vaskevicius, Narunas and Palmieri, Luigi and Chebrolu, Nived and Arras, Kai O and Stachniss, Cyrill},
  booktitle={2023 IEEE/RSJ International Conference on Intelligent Robots and Systems (IROS)},
  pages={11218--11225},
  year={2023},
  organization={IEEE}
}

@INPROCEEDINGS{11312999,
  author={Massai, Leonardo and Ferrari-Trecate, Giancarlo},
  booktitle={2025 IEEE 64th Conference on Decision and Control (CDC)}, 
  title={Free Parametrization of L2-bounded State Space Models}, 
  year={2025},
  volume={},
  number={},
  pages={7012-7017},
  doi={10.1109/CDC57313.2025.11312999}}

@article{revay2023recurrent,
  title={Recurrent equilibrium networks: Flexible dynamic models with guaranteed stability and robustness},
  author={Revay, Max and Wang, Ruigang and Manchester, Ian R},
  journal={IEEE Transactions on Automatic Control},
  volume={69},
  number={5},
  pages={2855--2870},
  year={2023},
  publisher={IEEE}
}

@inproceedings{kon2023directlearningparametervaryingfeedforward,
  title={Direct learning for parameter-varying feedforward control: A neural-network approach},
  author={Kon, Johan and Van De Wijdeven, Jeroen and Bruijnen, Dennis and T{\'o}th, Roland and Heertjes, Marcel and Oomen, Tom},
  booktitle={2023 62nd IEEE Conference on Decision and Control (CDC)},
  pages={3720--3725},
  year={2023},

}

@ARTICLE{10633771,
  author={Furieri, Luca and Galimberti, Clara Lucía and Ferrari-Trecate, Giancarlo},
  journal={IEEE Open Journal of Control Systems}, 
  title={Learning to Boost the Performance of Stable Nonlinear Systems}, 
  year={2024},
  volume={3},
  number={},
  pages={342-357},
  doi={10.1109/OJCSYS.2024.3441768}}

@INPROCEEDINGS{11312988,
  author={Kirsch, Nicolas and Massai, Leonardo and Ferrari-Trecate, Giancarlo},
  booktitle={2025 IEEE 64th Conference on Decision and Control (CDC)}, 
  title={Boosting the transient performance of reference tracking controllers with neural networks}, 
  year={2025},
  volume={},
  number={},
  pages={8015-8021},
  doi={10.1109/CDC57313.2025.11312988}}

@inproceedings{heemels2012introduction,
  author    = {Heemels, W. P. M. H. and Johansson, Karl H. and Tabuada, Paulo},
  title     = {An introduction to event-triggered and self-triggered control},
  booktitle = {51st IEEE Conference on Decision and Control (CDC)},
  pages     = {3270--3285},
  year      = {2012}
}

@article{tabuada2007event,
  author  = {Tabuada, Paulo},
  title   = {Event-triggered real-time scheduling of stabilizing control tasks},
  journal = {IEEE Transactions on Automatic Control},
  volume  = {52},
  number  = {9},
  pages   = {1680--1685},
  year    = {2007}
}

@article{tallapragada2013event,
  author  = {Tallapragada, Pavankumar and Chopra, Nikhil},
  title   = {On event triggered tracking for nonlinear systems},
  journal = {IEEE Transactions on Automatic Control},
  volume  = {58},
  number  = {9},
  pages   = {2343--2348},
  year    = {2013}
}

@article{wang2011event,
  author  = {Wang, Xiaofeng and Lemmon, Michael D.},
  title   = {Event-triggering in distributed networked control systems},
  journal = {IEEE Transactions on Automatic Control},
  volume  = {56},
  number  = {3},
  pages   = {586--601},
  year    = {2011}
}

@article{ali2025optimal,
  title={Optimal micro-grid battery scheduling within a comprehensive smart pricing scheme},
  author={Ali, Mohammed Ashraf and Besheer, Ahmad H and Emara, Hassan M and Bahgat, Ahmed},
  journal={Scientific Reports},
  volume={15},
  number={1},
  pages={20098},
  year={2025},
  publisher={Nature Publishing Group UK London}
}

@inproceedings{yen2000intelligent,
  title={Intelligent fault tolerant control using artificial neural networks},
  author={Yen, Gary G and Ho, Liang-Wei},
  booktitle={Proceedings of the IEEE-INNS-ENNS International Joint Conference on Neural Networks. IJCNN 2000. Neural Computing: New Challenges and Perspectives for the New Millennium},
  volume={1},
  pages={266--271},
  year={2000},
}

@inproceedings{zimmerman2014neural,
  title={Neural network based obstacle avoidance using simulated sensor data},
  author={Zimmerman, Timothy A},
  booktitle={2014 ASEE Zone 1 Conference},
  year={2014}
}

@article{bonzanini2024perception,
  title={Perception-aware model predictive control for constrained control in unknown environments},
  author={Bonzanini, Angelo D and Mesbah, Ali and Di Cairano, Stefano},
  journal={Automatica},
  volume={160},
  pages={111418},
  year={2024},
  publisher={Elsevier}
}

@INPROCEEDINGS{8593739,
  author={Falanga, Davide and Foehn, Philipp and Lu, Peng and Scaramuzza, Davide},
  booktitle={2018 IEEE/RSJ International Conference on Intelligent Robots and Systems (IROS)}, 
  title={{PAMPC}: Perception-Aware Model Predictive Control for Quadrotors}, 
  year={2018},
  volume={},
  number={},
  pages={1-8},
  doi={10.1109/IROS.2018.8593739}}

@article{kirsch2025resilient,
  title={Resilient {AFE} Drive Control using Neural Networks with Tracking Guarantees},
  author={Kirsch, Nicolas and Arghir, Catalin and Mastellone, Silvia and Ferrari-Trecate, Giancarlo},
  journal={arXiv preprint arXiv:2512.03545},
  year={2025}
}

\end{document}